\documentclass[11pt]{article}

\usepackage{booktabs}
\usepackage[utf8]{inputenc}

\usepackage{amsmath,amsfonts,amsthm,mathrsfs,xspace,graphicx}
\usepackage[backref,colorlinks,bookmarks=true, citecolor=blue, urlcolor=black]{hyperref}
\usepackage{endnotes}
\usepackage{enumerate} 
\usepackage{color}
\usepackage{float}
\usepackage{xcolor}
\definecolor{since}{rgb}{0.5,0.5,0.5}
\definecolor{newred}{HTML}{ED2024}
\definecolor{newgreen}{HTML}{109A48}
\definecolor{newblue}{HTML}{535DAA}
\definecolor{neworange}{HTML}{F79420}
\usepackage{mdframed}
\usepackage{bbm}
\usepackage{suffix} \usepackage{times}
\usepackage{tabularx}
\usepackage{makecell}
\usepackage{amssymb,latexsym}
\usepackage[capitalize]{cleveref}
\usepackage[font=footnotesize,labelfont=bf]{caption}
\usepackage{enumitem}
\usepackage{tikz}
\usepackage{tikz-cd}
\usepackage{multirow}

\usepackage[section]{placeins}
\usepackage[affil-it]{authblk}

\makeatletter
\renewcommand*\env@matrix[1][*\c@MaxMatrixCols c]{%
  \hskip -\arraycolsep
  \let\@ifnextchar\new@ifnextchar
  \array{#1}}
\makeatother

\usepackage[margin=1in]{geometry}

\usepackage{endnotes}
\usepackage{amssymb,latexsym}
\usepackage{enumitem}
\setlist{itemsep=0mm}
\usepackage{tikz}
\usepackage{float}
\usepackage{setspace}
\usepackage{soul}

\usepackage{thmtools}
\usepackage{thm-restate}

\usepackage{mathtools, physics}
\usepackage{stmaryrd}
\usepackage{complexity}

\newclass{\QPCP}{QPCP}
\newclass{\QCPCP}{QCPCP}
\newclass{\QCMAcomp}{QCMA-complete}
\newclass{\sharpP}{\#P}

\usepackage[linesnumbered,ruled,vlined]{algorithm2e}
\usepackage{verbatim}
\usepackage{subcaption}
\usepackage{qcircuit}

\usepackage{amsthm}
\newtheorem{theorem}{Theorem}
\newtheorem*{theorem*}{Theorem}

\newtheorem*{proposition*}{Proposition}
\newtheorem{fact}[theorem]{Fact}
\newtheorem*{fact*}{Fact}
\newtheorem{lemma}[theorem]{Lemma}
\newtheorem*{lemma*}{Lemma}

\newtheorem{corollary}[theorem]{Corollary}

\newtheorem*{conjecture*}{Conjecture}

\theoremstyle{definition}
\newtheorem{definition}[theorem]{Definition}
\newtheorem*{definition*}{Definition}

\theoremstyle{remark}

\newtheorem*{remark*}{Remark}

\DeclareMathOperator\eye{\mathbb{I}}

\usepackage{accents}

\newcommand\restr[2]{{
  \left.\kern-\nulldelimiterspace 
  #1 
  \right|_{#2} 
  }}
  
\SetKwInput{KwInput}{Input}                %
\SetKwInput{KwOutput}{Output}

\newcommand{\mc}{\mathcal}
\renewcommand{\E}{\mathop{\mathbb E\/}}

\title{Polynomial-Time Classical Simulation of Noisy Quantum Circuits with Naturally Fault-Tolerant Gates}
\author[1,2]{\href{}{Jon~Nelson*}}
\author[1,2]{\href{https://orcid.org/0000-0001-6365-8238}{Joel~Rajakumar*}}
\author[3,1]{\href{}{Dominik Hangleiter}}
\author[1,2]{\href{}{Michael J. Gullans}}

\affil[1]{\normalsize  Joint Center for Quantum Information \& Computer Science, University of Maryland and NIST}
\affil[2]{\normalsize Department of Computer Science,
	University of Maryland}
\affil[3]{\normalsize Simons Institute for the Theory of Computing, University of California at Berkeley}
\affil[4]{\normalsize National Institute of Standards and Technology (NIST)}

\begin{document}

\maketitle
\def\thefootnote{*}\footnotetext{These authors contributed equally to this work}\def\thefootnote{\arabic{footnote}}
\begin{abstract}
    We construct a polynomial-time classical algorithm that samples from the output distribution of noisy geometrically local Clifford circuits with any product-state input and single-qubit measurements in any basis. Our results apply to circuits with nearest-neighbor gates on an $O(1)$-D architecture with depolarizing noise after each gate. Importantly, we assume that the circuit does not contain qubit resets or mid-circuit measurements. This class of circuits includes Clifford-magic circuits and Conjugated-Clifford circuits, which are important candidates for demonstrating quantum advantage using non-universal gates. Additionally, our results can be extended to the case of IQP circuits augmented with CNOT gates, which is another class of non-universal circuits that are relevant to current experiments. Importantly, these results do not require randomness assumptions over the circuit families considered (such as anticoncentration properties) and instead hold for \textit{every} circuit in each class as long as the depth is above a constant threshold. This allows us to rule out the possibility of fault-tolerance in these circuit models. 
    As a key technical step, we prove that interspersed noise causes a decay of long-range entanglement at depths beyond a critical threshold. To prove our results, we merge techniques from percolation theory and Pauli path analysis. 
\end{abstract}
\section{Introduction}
A first step towards understanding the power of quantum computation is to determine the conditions under which it is possible to perform a quantum computation that cannot be classically simulated or, in other words, to demonstrate ``quantum advantage''. 
Although there is robust theoretical evidence that this is true for large-scale fault-tolerant quantum computers, this becomes a subtle question when restricted to near-term quantum hardware. 
These devices are noisy and may lack capabilities that are required for fault tolerance such as the ability to perform intermediate measurements or reset qubits during a computation. 

To address these questions, a line of research has considered the task of sampling from the output distribution of (random) quantum circuits as a way to demonstrate quantum advantage \cite{terhal_adaptive_2004}. 
This is relevant to near-term hardware because even very restricted models of quantum computation can be used to demonstrate such quantum advantages, and the advantage is not contrived in that it tolerates a constant amount of total error \cite{bremner_classical_2010,bremner_average-case_2016,aaronson_computational_2013,gao_quantum_2017,hamilton_gaussian_2017}.
Adding realistic noise with constant rate complicates this picture, since the fidelity with which the targeted task can be performed on a noisy quantum device diminishes exponentially with the system size.
It is therefore important to study the computational power of quantum computations under realistic noise assumptions~\cite{aharonov1996limitations,fujii_computational_2016,bremner_achieving_2017,aharonov_polynomial-time_2023,chen2022complexitynisq}. 
Obtaining a precise understanding of the computational power of noisy restricted models of computation can elucidate the ingredients that are necessary for near-term quantum advantage.

In this work, we consider quantum circuits with gate sets that are naturally fault-tolerant. 
These include Clifford circuits with state preparation and measurement in arbitrary product bases, and IQP circuits augmented with CNOT gates. In particular, this allows for perfect magic state inputs in the case of Clifford circuits and arbitrary diagonal non-Clifford gates in the case of IQP+CNOT circuits.  
The circuit classes we consider do not allow for intermediate measurement and classical feed-forward 
and are therefore not universal for quantum computation; however, they constitute circuits which are readily implementable in early-stage fault-tolerant quantum devices, which can execute non-universal gate sets transversally \cite{bluvstein_logical_2024,eastin_restrictions_2009}, making them good candidates for demonstration of quantum advantage.
In particular, these circuit classes include Clifford-magic, conjugated Clifford circuits and hypercube block IQP circuits, which have all been proven to be hard to approximately sample from in the nearly noiseless case assuming certain complexity-theoretic conjectures \cite{yoganathan_quantum_2019, bouland,hangleiter_fault-tolerant_2024,pashayan_estimation_2020}. 
Thus, it is important to understand the noise and depth regimes in which these circuits maintain their computational power.

We develop a classical algorithm for sampling from the output of these circuit families when they are subject to mid-circuit depolarizing noise. Applied to geometrically local circuits, our algorithm is efficient at circuit depths above a \emph{constant threshold} that depends on the noise rate and is efficient for \textit{all} circuits in this class. Our results demonstrate that in order to achieve scalable quantum advantage with realistically noisy circuits, an additional ingredient is required such as a universal set of mid-circuit gates, intermediate measurements, or a supply of fresh qubits. 

\subsection{Overview of Results}

We summarize our main result informally in the following theorem. 

\begin{theorem}[Informal]
    There exists an efficient randomized classical algorithm that approximately samples from the output distribution of a noisy quantum circuit $C$ with circuit-level local noise rate $\gamma$ in the following cases,
    \begin{enumerate}
        \item $C$ is a geometrically local Clifford-Magic, Conjugated Clifford, or IQP+CNOT circuit and $d \ge \Omega(\gamma^{-1} \log \gamma^{-1})$
        \item $C$ is a Clifford-Magic, Conjugated Clifford, or IQP+CNOT circuit and $d \ge \Omega(\gamma^{-1}\log n)$ .
        \end{enumerate}
\end{theorem}

We use the term `geometrically local' to refer to circuits whose gates are nearest-neighbor on an $O(1)$-D lattice. Our algorithm has many desirable properties:
(1) it does not require assumptions on the output distribution or distribution from which the circuit is sampled. In particular, in contrast to the results of \cite{bremner_achieving_2017,gao_efficient_2018,aharonov_polynomial-time_2023,schuster_polynomial-time_2024} it does not require that the output distributions have the anticoncentration property (see \cite{hangleiter_computational_2023} for details on this property) and it applies to \textit{worst-case} circuits. 
(2) It performs \textit{exact} sampling with random runtime which is polynomial in expectation (and this can be modified to give the more common \textit{approximate} sampler with worst-case polynomial runtime).  (3) Only the depth threshold but not the runtime  depends on the noise strength.

Next, our proof techniques also show that noisy Clifford circuits acting on a random input bit string anticoncentrate in $O(\gamma^{-1} \log n)$ depth, which is proved in \cref{app:anticoncentration}. This implies that noisy random quantum circuits on arbitrary architectures anticoncentrate in $O(\log n)$-depth, which was only previously known for $1D$ and all-to-all connectivity \cite{dalzell_random_2022,deshpande_tight_2022} \footnote{recent concurrent work \cite{schuster_random_2024} also addresses this gap for the noiseless case}.
Our proof is conceptually significant because it demonstrates that anticoncentration can arise from the noise alone, rather than due to the randomness of the gates (in the noiseless setting, these circuits can be non-anticoncentrated).

\subsection{Related Works}

The strongest upper bound on the computational complexity of worst-case noisy circuits is due to Aharonov \textit{et al.} \cite{aharonov1996limitations}, who showed that the output distribution of any noisy quantum circuit converges to uniform after circuit depth which is logarithmic in system size, and is thus trivially simulatable. The same work \cite{aharonov1996limitations} also demonstrates a lower bound on the computational complexity of general quantum circuits by proving that noisy quantum circuits can solve problems in $QNC^1$ (decision problems solvable by noiseless logarithmic depth quantum circuits) with a quasipolynomial overhead in system size.
At the same time, Shor’s factoring algorithm can be implemented using log-depth quantum circuits \cite{cleve_fast_2000}.  
Together, these results provide some evidence that it is difficult in general to classically simulate noisy quantum circuits at logarithmic depths.
This suggests to start by understanding the classical simulatability of restricted families of noisy quantum circuits that correspond to the limitations of near-term hardware, such as geometric locality or limited gate sets. 

Various prior works have provided classical simulation algorithms for the setting of noisy Clifford circuits with single-qubit non-Clifford gates ~\cite{fujii_computational_2016,Virmani_2005,kempe2008upperboundsnoisethreshold,Plenio_2010,seddon_quantifying_2021,gonzálezgarcía2024paulipathsimulationsnoisy}. When applied to Clifford-Magic or Conjugated Clifford circuits, these algorithms are only efficient when the depolarizing noise rate is above $\sim 15\%$. Our algorithm expands this classical simulability regime by showing that circuits with \textit{any} constant noise rate are efficiently simulable as long as the depth is above a constant threshold.
Our algorithm is also conceptually quite different. In prior approaches, the noise converts a circuit containing non-Clifford (`magic') gates into a probabilistic mixture of Clifford circuits, which can then be simulated using the Gottesman-Knill theorem. 
In contrast, our algorithm leverages the loss of entanglement, rather than the loss of magic, as the key to classical simulatability. Specifically, the noise transforms a circuit with entangling gates into a mixture of circuits with a tensor-product form, which can be simulated via brute-force methods \footnote{Note, each circuit in the mixture can contain up to $\Omega(n)$ non-Clifford gates, which means Gottesman-Knill-type algorithms would likely fail unless the tensor product structure is exploited.}. Furthermore, we carefully account for the accumulation of noise in the circuit, and find that increasing the depth reduces the entanglement even further, allowing for classical simulability at lower noise strengths. At first glance, this is somewhat unintuitive because one might expect higher depth circuits to possess more entanglement and to be more robust to noise. Our results demonstrate an important idea: as depth increases, the decohering effects of accumulating noise always win out against the entangling/fault-tolerance abilities of quantum gates.

A rich landscape of existing work has also considered noisy circuits where the gates are chosen randomly ~\cite{aharonov_polynomial-time_2023,bremner_achieving_2017,TAKAHASHI2021117, gao_efficient_2018, mele2024noiseinducedshallowcircuitsabsence, schuster_polynomial-time_2024, fontana2023classicalsimulationsnoisyvariational}. The accuracy of these algorithms strongly relies on the randomness property and thus does not exclude the possibility of a \textit{worst-case} circuit which is cleverly designed to maintain robustness to noise (e.g. \cite{aharonov1996limitations,fujii_computational_2016}). Therefore, these results do not have implications for early fault-tolerance. The work of \cite{gonzálezgarcía2024paulipathsimulationsnoisy} and \cite{stilck_franca_limitations_2021} consider worst-case circuits; however, these results are incomparable since they only apply in the easier setting of estimating expectation values rather than sampling from the output distribution. 
The only prior work achieving both properties (worst-case and sampling before $\log(n)$ depth) is Refs.~\cite{rajakumar_polynomial-time_2024, oh_classical_2024}, which take advantage of commutation properties of the gate set to prove their bounds. In contrast, our results and our proof techniques suggest that classical simulatability can arise from more general properties such as unitarity and geometric locality.

We also note that phase transitions in entanglement properties of quantum circuits due to noise/percolation have been observed in a variety of other settings (e.g. \cite{Suzuki_2025,aharonov_quantum_2000,browne_phase_2008-1,fujii_computational_2016}); however, a key distinction in our results is that we demonstrate a phase transition in depth for any noise strength $\gamma = \Omega(1)$ (not just noise which is greater than a critical threshold).

\subsection{Proof Techniques} 

\begin{figure}
    \centering
    \includegraphics[width=0.95\columnwidth]{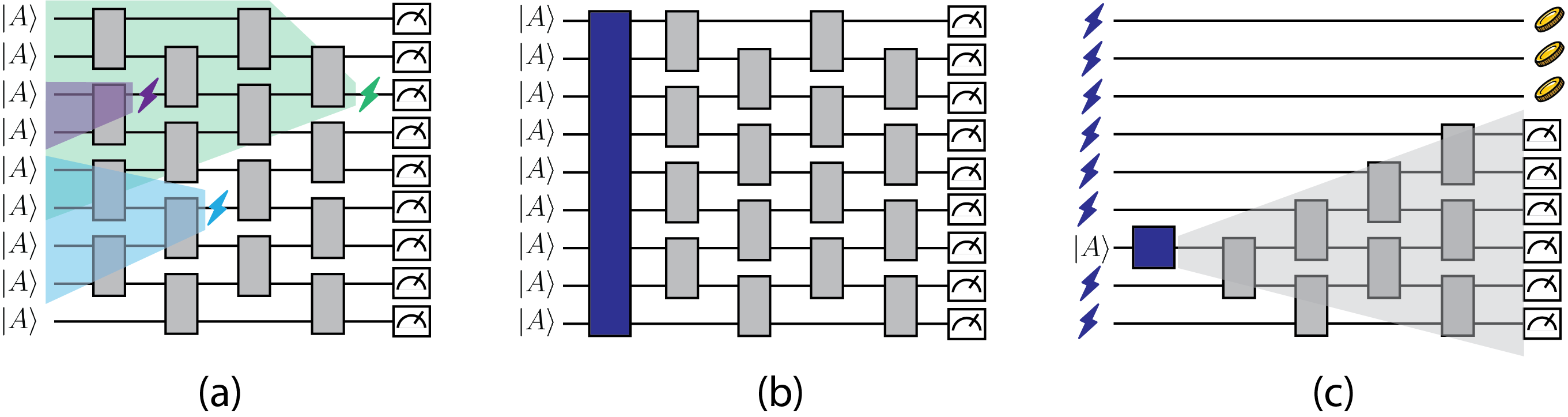}
    \caption{Overview of the simulation algorithm. (a) Depolarizing errors are sampled and propagated to the beginning of the circuit. $\ket{A}$ represents a single-qubit magic-state but can be any single-qubit state. (b) The circuit can now equivalently be represented as one error channel followed by the original noiseless circuit. (c) We prove that this input error channel has the effect of depolarizing many of the input qubits and so it remains to simulate the lightcones of qubits that are not depolarized. When these lightcones intersect they must be simulated together but can otherwise be simulated independently. Any measurements that are not in the lightcone of a depolarized qubit can be simulated by a random coin flip.}
    \label{fig:diagram}
\end{figure}

At a high-level, the algorithm works by first propagating errors to the beginning of the circuit. This reformulates the noisy circuit as one layer of noise followed by the ideal circuit, which is pictured in \cref{fig:diagram}b. 
The focus of much of the technical work is to show that this propagated error channel (which we denote by the channel $\Pi$) effectively depolarizes many of the input qubits, which is pictured in \cref{fig:diagram}c. This is similar to \cite{rajakumar_polynomial-time_2024, oh_classical_2024}; however, we require new tools to deal with the fact that the noise spreads out non-trivially, rather than commuting through the gates. In particular, we analyze which input Pauli operators can survive the noise. First, we show that 
\begin{quote}
    (1) It is exponentially unlikely (in the depth and the number of qubits) that any input Pauli operator acting on a large island of qubits survives $\Pi$.
\end{quote}
    We prove this statement using a careful counting argument, in place of assuming anticoncentration as done in prior work \cite{bremner_achieving_2017,gao_efficient_2018,aharonov_polynomial-time_2023}. While the Cliffordness of the circuit simplifies the Pauli propagation significantly, the main ingredient of the proof is the more general property of \textit{unitarity} to argue that not all high-weight Pauli operators can be mapped by the circuit to low-weight operators since this is a strictly smaller space and thus violates unitarity. This suggests that high-weight Pauli operators maintain their approximate weight throughout the circuit and thus are evenly suppressed by errors at various timesteps, i.e. they experience an exponential decay in both depth and weight. The next step of the analysis requires the depth to be large enough so that each Pauli operator experiences enough noise to suppress the exponential number of possible high-weight Pauli operators within an island of qubits. It turns out this occurs at a constant depth of around $d^* = O(\gamma^{-1}\log{\gamma^{-1}})$, which can be understand as a percolation effect that occurs due to \textit{geometric locality}. We then use the probabilistic method to show that 
     \begin{quote}
         (2) If a large island of qubits was not depolarized by $\Pi$, then there must be at least one Pauli operator with support on these qubits, which survives $\Pi$.
     \end{quote}
     Since we proved in (1) that the survival of large input Pauli operators is extremely unlikely under $\Pi$, then it must also be true by (2) that no large islands of qubits remain non-depolarized by $\Pi$.

     Having proved this result on the entanglement structure of the noisy state, we can now simulate the circuit by considering each island of qubits separately. If we were to simulate these qubits by brute-force, then this would only allow for a quasipolynomial time algorithm when $D >1$ for $\Omega(\log n)$-depth. However, we can improve this dramatically by decomposing the input state in the Pauli basis and only simulating Pauli operators that survive the channel - this is why our algorithm also works when there is all-to-all connectivity and no lightcones, although this requires the larger depth.

\subsection{Discussion and Future Work}
\label{sec:discussion}

\subsubsection{Naturally Fault-Tolerant Gate Sets} 

We consider the Clifford and IQP+CNOT gate sets to be `naturally' fault-tolerant because they are transversal for certain quantum error-correcting codes such as 3D color codes \cite{bombin2015gaugecolorcodesoptimal}, and thus can be implemented fault-tolerantly with relative ease. These facts have been exploited in various experimental implementations \cite{bluvstein_logical_2024,hangleiter_fault-tolerant_2024}. In addition, Clifford circuits, in particular, are used in many important error correction gadgets. For instance, the encoding and syndrome extraction circuits for any stabilizer code can be implemented as a Clifford circuit. Furthermore, Pauli noise propagates predictably through Clifford gates, which allows for a decoding algorithm to track the error at various times in the circuit in order to diagnose the error. In fact, there have been a variety of proposals for early fault-tolerance \cite{bravyi_quantum_2020,bergamaschi_liu,delfosse2023spacetimecodescliffordcircuits} involving propagation of interspersed Pauli errors through Clifford circuits where they can be corrected at the end of the computation, provided that the final measurement contains some syndrome information about the space-time location of the error. This is a crucial part of the techniques in \cite{bravyi_quantum_2020}, which demonstrates quantum advantage over shallow classical circuits in the presence of noise. These desirable error correction properties have led to the proposal of using only Clifford gates on magic-state input for achieving fault-tolerant universal quantum computation \cite{bravyi_universal_2005}. The crucial difference in our setting is that we do not allow for feed-forward adaptive operations, which is responsible for the universality.

\subsubsection{Noisy Quantum Advantage without Intermediate Measurement/Reset}

Our results rule out certain proposals for quantum advantage in noisy circuits without intermediate measurement or reset operations (which would otherwise enable fault-tolerance through the threshold theorem). 
In particular, one might hope that it would be possible to construct some clever Clifford/IQP+CNOT circuit to demonstrate quantum advantage, where resource states are prepared perfectly, and interspersed errors are corrected using final measurement data (e.g. like \cite{yunchao_quynh} or \cite{bravyi_quantum_2020}). However, our results imply that such circuits cannot be deeper than constant depth when they are implemented under geometric locality and gate set constraints.

\subsubsection{Computational Complexity Phase Transition at Constant Depth}

We further conjecture that the depth threshold for classical simulatability in the geometrically local case corresponds directly to a phase transition in computational complexity (as demonstrated in the IQP case in \cite{rajakumar_polynomial-time_2024}). That is, it may be possible to construct hard-to-sample, geometrically local quantum circuits at depths right below the threshold (at least for exact sampling). Indeed, a primary motivation for developing asymptotically efficient classical simulation algorithms is to narrow down the regime of near-term quantum advantage, and thereby point towards candidates for noise-robust quantum advantage. By process of elimination, our results point to the possibility that Clifford-magic/IQP+CNOT circuits could evade classical simulatability at depths below $O(1)$. Such a construction could have important consequences for demonstrating near-term quantum advantage. \cite{bergamaschi_liu} provides some evidence in favor of this conjecture by demonstrating that constant-depth geometrically local noisy Clifford circuits possess long-range entanglement.

\subsubsection{Peaked Quantum Circuits}
Our work is among the first classical sampling algorithms for noisy quantum circuits that apply to worst-case circuits, albeit with restrictions on connectivity and gate sets. This is accomplished by removing the need for the anticoncentration property in our analysis. This not only allows for classical simulatability beyond average case, but also shows that our algorithm works even when the circuit output distribution is bounded away from uniform (the distributions we can simulate may be peaked). Peaked circuits have recently been proposed as a potential candidate for verifiable quantum advantage \cite{aaronson_verifiable_2024} and are by definition not amenable to simulation techniques relying on anticoncentration. Our methods do not have the same fundamental restrictions and so we are able to rule out noisy quantum advantage for a certain class of candidate circuits with peaked distributions\footnote{Note that it is known that constant-depth peaked circuits can be simulated in quasipolynomial time \cite{bravyi_classical_2023} but it is an open question whether they are simulatable at $O(\log n)$-depth and/or with noise.}.

\subsubsection{Generality of Percolation} \label{sec:generality}
Another ingredient in our simulation result that may apply more generally is the existence of a constant depth threshold for noise percolation in geometrically local circuits. In recent prior work, noise percolation was shown to exist in noisy IQP circuits \cite{rajakumar_polynomial-time_2024}, which also resulted in classical simulation. For these IQP circuits, the analysis relies on the special property that commuting gates cannot spread entanglement very quickly (i.e. each qubit's lightcone grows linearly with depth) and that certain noise channels commute with the gate set. This may give reason to believe that this percolation effect is unique to IQP circuits. However, our results show that in more general settings, such as Clifford circuits, where the noise can spread in much more complicated ways, percolation still occurs. In these cases, the cause of percolation can be attributed instead to geometric locality, which may be a more fundamental limitation to near-term quantum circuits.

\subsubsection{Current Experiments}
Our results may also provide avenues to classically simulate real quantum experiments. For example, going beyond geometric locality, can we use percolation to prove a sublogarithmic depth threshold for classical simulatability on partially constrained architectures, such as the hypercube connectivity of \cite{hangleiter_fault-tolerant_2024}, which was recently implemented experimentally \cite{bluvstein_logical_2024}? Note that the IQP+CNOT circuits of \cite{hangleiter_fault-tolerant_2024} require $O(\log n)$ depth. Our algorithm for the all-to-all connectivity case shows that these circuits will be classically simulatable if they are scaled past a $O(\gamma^{-1}\log n)$ depth (for noise strength $\gamma$). It would be of practical interest to see if this bound can be improved by considering the locality structure of the hypercube.


\subsection{Paper Outline} In \cref{sec:prelims} we give the preliminaries and notation. In \cref{sec:main} we describe the algorithm. In \cref{sec:relation}, we describe how these results directly apply to IQP+CNOT circuits and to conjugated Clifford circuits.



\section{Preliminaries} \label{sec:prelims}

\subsection{Pauli and Clifford Group}

We define the Pauli group as:

\begin{align*}
        \mathsf{P}_n := \{1,i,-1,-i\} \times \{\eye,X,Y,Z\}^{\otimes n}
    \end{align*}
We define the phaseless Pauli group $\hat{\mathsf{P}}_n \leqslant \mathsf{P}_n$ as:
    \begin{align*}
        \hat{\mathsf{P}}_n := \{1\} \times \{\eye,X,Y,Z\}^{\otimes n}
    \end{align*}   Let $\mathsf{P}^Z_n$ denote the subgroup of $\mathsf{P}^Z_n$ where each single-qubit operator is in $\{\eye,Z\}$. Let $\mathsf{P}^X_n$ be defined similarly.  Let $X_i \in \mathsf{P}_n$ denote the n-qubit Pauli operator where $X$ acts on qubit $i$ and $\eye$ acts elsewhere and define $Z_i$ analogously. The Clifford group is defined as:
    \begin{align*}
        \mathsf{C}_n := \{ U \in \mathsf{U}(2^n): UPU^\dagger \in \mathsf{P}_n \forall P \in \mathsf{P}_n\}
    \end{align*}
    where $\mathsf{U}(2^n)$ is the group of $2^n \times 2^n$ unitary matrices.

    

\subsection{Binary symplectic representation and Tableau matrix}
The binary symplectic representation is a mapping $v: \hat{\mathsf{P}}_n \rightarrow \mathbb Z^{2n}_2$ where the first $n$ bits represent the $X$-components and the second $n$ bits represent the $Z$-components of the corresponding Pauli operator. More precisely, for any $s \in \hat{\mathsf{P}}_n$, if the $i$th qubit of $s$ contains $X$ or $Y$ then the $i$th bit of $v(s)$ is $1$, otherwise it is 0. Similarly, if the $i$th qubit of $s$ contains $Z$ or $Y$ then the $(i+n)$th bit of $v(s)$ is $1$, otherwise it is 0. The symplectic inner product between two vectors $u$ and $w$ is defined as $u ^T \Lambda w$, where
\begin{align*}
    \Lambda =
        \begin{pmatrix}[cc]
                    0 & \eye_n \\
                    \eye_n & 0
        \end{pmatrix}
    \end{align*}
Notice that $u ^T \Lambda w= 0$ if and only if the corresponding Pauli operators commute.

Given a set of Pauli operators $M \subseteq \mathsf{P_n}$, we define $ \mathrm T_M$ to be the $|M| \times 2n$ binary matrix where each row represents the binary symplectic vector corresponding to a Pauli operator in $M$. We call this the `tableau matrix'. We will later make use of the fact that the nullspace of $\mathrm T_M\Lambda$ over $\mathbb Z_2$ exactly corresponds to the subgroup of all Pauli operators that commute with all elements of $M$.

\subsection{Channels and Unitaries}

We use $C$ to refer to Clifford unitaries and $\mathcal{C}$ to refer to their corresponding channel $\mathcal{C}(\cdot) = C (\cdot) C^\dagger$. Given a Clifford unitary $C$ that can be decomposed into a sequence of layers $C = U_d \dots U_1$, we denote the subsequence of unitaries up to timestep $i$ as:
\begin{align}
    C_{i} :=U_i \dots U_1
\end{align}
Our results apply for single-qubit depolarizing channels:
\begin{align}
    \mathcal{E}(\rho) = (1-\gamma)\rho + \gamma\frac{\eye}{2}\Tr(\rho)
\end{align}
Throughout, $\Tr_A$ denotes the partial trace over qubits in $A$ and $\Tr_{-A}$ denotes the partial trace over all qubits not in $A$.
We will often refer to noisy Clifford circuits as $ \Tilde{\mc C} = \mathcal{E}^{\otimes n} \mathcal{U}_d \mathcal{E}^{\otimes n} \dots \mathcal{U}_1 \mathcal{E}^{\otimes n} $, where each $\mathcal{U}_i$ is a layer of the Clifford circuit.

\subsection{Output distributions}
If $\{F_x\}_x$ is a POVM representing the measurement basis, $\rho$ is an initial state, and $\Tilde{\mc C}$ is a channel representing a noisy circuit, we will refer to the resulting output distribution as $p_{\tilde{\mc C}}$,
\begin{align}
     p_{\tilde{\mc C}}(x) = \Tr(F_x \Tilde{\mc C}(\rho))
\end{align}
Moreover, we use TVD to describe how close two distributions are (i.e. approximation error),
\begin{equation}
     \Delta(p,q):=\left\|p-q\right\|_1 := \sum_{x\in\{0,1\}^n}\left|p(x)-q(x)\right|
\end{equation}

\section{Main Results} \label{sec:main}

\begin{theorem} \label{theorem:main}
    Suppose $\mc C$ is a Clifford circuit of depth $d$ on $n$ qubits. Let $\tilde{\mc C}$ denote the noisy implementation of $\mc C$, where each layer is interspersed with depolarizing channels $\mathcal{E}$ of strength $\gamma$ on every qubit. Let $p_{\tilde{\mc C}}$ be the output distribution obtained from preparing $n$ qubits in an arbitrary product state, applying $\tilde{\mc C}$, and measuring each qubit in an arbitrary basis. There exists a depth threshold $d^* = O(\gamma^{-1} \log n)$ such that when $d > d^*$, there exists a randomized classical algorithm that exactly samples from $p_{\tilde{\mc C}}$ with random runtime $T$ of expected value,
    \begin{align}
        E[T] = O(\poly(n))
    \end{align}
    If $\mc C$ is further restricted to nearest-neighbor gates on a $D$-dimensional lattice, there exists a constant depth threshold $d_{local}^* = O(\gamma^{-1}3^{2D}+ D\gamma^{-1}\log(D\gamma^{-1}))$ such that when $d > d_{local}^*$, 
    \begin{align}
        \E [T] =O(d^{3+D}n^4)
    \end{align}
\end{theorem}

Using a standard reduction (applying Markov's inequality), we can convert this `Las Vegas' algorithm into a `Monte Carlo' algorithm with polynomial worst-case runtime. 

\begin{corollary} \label{corollary:monte-carlo}
Using the same notation as \cref{theorem:main}, there exists a depth threshold $d^* = O(\gamma^{-1} \log n)$ such that when $d > d^*$, there exists a randomized classical algorithm that samples from $q_{\tilde{\mc C}}$, such that $\|q_{\tilde{\mc C}} - p_{\tilde{\mc C}}\|_{TVD} \leq \epsilon$, with runtime
    \begin{align}
        T = O(\epsilon^{-1}\poly(n))
    \end{align}
If $\mc C$ is further restricted to nearest-neighbor gates on a $D$-dimensional lattice, there exists a constant depth threshold $d_{local}^* = O(\gamma^{-1}3^{2D}+D\gamma^{-1}\log(D\gamma^{-1}))$ such that when $d > d_{local}^*$, 
\begin{align}
    T = O(\epsilon^{-1}d^{3+D}n^4)
\end{align}
\end{corollary}

\subsection{Preprocessing the Noisy Circuit} \label{sec:preprocessing}
In this section, we outline a few techniques we use to simplify the noisy circuit, and make it more amenable to classical simulation.

\subsubsection{Stochastic Application of Noise Channels} \label{sec:stochastic}
A key idea is to view each depolarizing channel in the circuit as a stochastic process that traces out the target qubit and replaces it with the identity with probability $\gamma$ (which we refer to as a `depolarizing error'), and does nothing otherwise. This stochastic method of applying the depolarizing channel is exactly equivalent in expectation to the more common method of treating it as a deterministic CPTP map. To formalize this, we define the following, 
\begin{definition} \label{def:main}
    Suppose $ \Tilde{\mc C} = \mathcal{E}^{\otimes n} \mathcal{U}_d \mathcal{E}^{\otimes n} \dots \mathcal{U}_1 \mathcal{E}^{\otimes n} $ (where each $\mathcal{U}_t$ is a layer of two-qubit Clifford gates). Let $\mc L$ represent the ensemble of error locations.
    For any $b$ drawn from $\mc L$, we use $\tilde{\mc C}_b$ to denote the channel that corresponds to applying a depolarizing error at each of the locations specified by $b$ in $\mc C$.
\end{definition}
Note by definition,
\begin{align}
    \tilde{\mc C}(\rho) = \E_b \tilde{\mc C}_b(\rho)
\end{align}
where $\E_b$ is used throughout the paper as shorthand for $\E_{b \sim \mc L}$. Note, $p_{\tilde{\mc C}}(x) = \E_b p_{\tilde{\mc C}_b}(x)$ because the expectation is a linear operator that commutes with any POVM element. Thus, we can exactly sample from $p_{\tilde{\mc C}}$ by first sampling $b$ and then exactly sampling the measurement outcomes of $\tilde{\mc C}_b$.

\subsubsection{Error Propagation through Clifford Circuits} \label{sec:local}

First, we sample a configuration of depolarizing errors $b \sim \mc L$ as described in \cref{sec:stochastic}. We will now show how to sample from $p_{\tilde{\mc C}_b}$. A crucial step of our algorithm is to convert the depolarizing errors in $\tilde{\mc C}_b$ to an error channel that acts only on the circuit input. To this end, we introduce the following notation,
\begin{definition}[Pauli projection channel]
\begin{align}
    \Pi_P(\rho) &= \frac{1}{2}\rho + \frac{1}{2}P \rho P^\dagger
\end{align}
where $P \in \mathsf{P}_n$ given an $n$-qubit state $\rho$.
\end{definition}
Note since $\Pi_{X} \circ  \Pi_{Z} (\rho) = \frac{1}{4} \rho + \frac{1}{4} X\rho X^\dagger + \frac{1}{4} Y\rho Y^\dagger + \frac{1}{4} Z\rho Z^\dagger = \eye/2$, the depolarizing errors in $\tilde{\mc C}_b$ can be simulated by the application of $\Pi_{X} \circ \Pi_{Z}$. The advantage of simulating depolarizing errors by applying these Pauli projection channels is that they can be propagated to the beginning of the Clifford circuit. For instance given a layer of Clifford gates $\mc U$, then
    \begin{align}
        \Pi_P \circ \mathcal{U} (\rho)  &= \frac{1}{2} U \rho U^\dagger + \frac{1}{2} P  U \rho U^\dagger P^\dagger \\
         &= \frac{1}{2} U \rho U^\dagger + \frac{1}{2} U U^\dagger P  U \rho U^\dagger P^\dagger U U^\dagger \\
        &=  \mc U \circ \Pi_{\mc U^\dagger(P)} (\rho)
    \end{align}
This motivates us to define the following,
\begin{definition} [Error Propagation]
    Let $M_b$ denote the set of Pauli operators generated by replacing depolarizing errors in $\tilde{\mc C}_b$ with Pauli projection channels and propagating these to the beginning of the circuit. Define the corresponding channel $\Pi_{M_b}(\rho) :=\bigcirc_{P \in M_b} \Pi_{P} (\rho)$ acting on the circuit input.
\end{definition}
Note we have shown that $\tilde{\mc C}_b(\rho) = \mc C \circ \Pi_{M_b}(\rho)$. Now, letting $\langle M_b \rangle \leqslant \mathsf{P}_n$ denote the subgroup generated by $M_b$ we can alternatively write $\Pi_{M_b}(\rho)$ as,
\begin{align} \label{eq:errorgroup}
    \Pi_{M_b}(\rho) =\bigcirc_{P \in M_b} \Pi_{P} (\rho)= \frac{1}{|\langle M_b \rangle|}\sum_{P \in \langle M_b \rangle}P \rho  P^\dagger
\end{align}
Crucially, the right-hand side of Eq. \ref{eq:errorgroup} only depends $\langle M_b \rangle$ rather than the generators $M_b$ themselves. Thus, this channel is equivalent to any composition of Pauli projection channels where the associated Pauli operators also generate $\langle M_b \rangle$. We occasionally refer to $\langle M_b \rangle$ as the ``error group".

\subsubsection{Action of Propagated Errors on Pauli Operators}
Note that the input state can always be written as a sum of Pauli operators using the Pauli basis and so it will be insightful to characterize the action of $\Pi_{M_b}$ on Pauli operators. This is described by the following fact.
\begin{fact}  \label{lemma:survivingpaulis}
    Let $s \in \hat{\mathsf{P}}_n$. If $s$ commutes with every Pauli operator in $M_b$, then $\Pi_{M_b}(s) = s$. Otherwise, $\Pi_{M_b}(s) = 0$
\end{fact}
\begin{proof}
    Recall that by \cref{eq:errorgroup}, $\Pi_{M_b}(s) = \bigcirc_{P \in M_b} \Pi_{P} (s)$. If $s$ commutes with all elements of $M_b$ then we have that for any $P \in M_b$
    \begin{align}
        \Pi_P(s) = \frac{1}{2}(s + PsP^\dagger) = \frac{1}{2}(s + s) = s
    \end{align}
    Thus, $\bigcirc_{P \in M_b} \Pi_{P} (s) = s$.
    Otherwise, there exists some $P \in M_b$ such that $P$ anti-commutes with $s$. In this case, 
    \begin{align}
        \Pi_{P}(s) = \frac{1}{2}(s + PsP^\dagger) = \frac{1}{2}(s - s) = 0
    \end{align}
    and so $\bigcirc_{P \in M_b} \Pi_{P} (s) = 0$.
\end{proof}

The set of Pauli operators that are preserved by the propagated error channel form a Pauli subgroup that is precisely the centralizer of $\langle M_b \rangle$, which we denote as $\mathsf{C}(\langle M_b \rangle)$. Recall that the centralizer of a group is the set of elements that commute with each element of the group. Using the tableau representation, $\mathsf{C}(\langle M_b \rangle)$ exactly corresponds to the nullspace of $\mathrm T_{M_b}\Lambda$.


\subsubsection{Effectively Depolarized Qubits}

We can now determine if an input qubit is depolarized by $\Pi_{M_b}$ by inspecting the error group $\langle M_b \rangle$. In particular, we have that if $X_i,Z_i \in \langle M_b\rangle $ then 
\begin{align}
    \label{eq:xiziimpliesdepol}
    \Pi_{M_b}(\rho) =\Pi_{M_b} \circ \Pi_{X_i} \circ \Pi_{Z_i} (\rho) = \Pi_{M_b} \left(\frac{\eye_i}{2} \otimes \Tr_i(\rho)\right)
\end{align}
where the first equality follows because $X_i,Z_i \in \langle M_b\rangle$ implies $\langle M_b \rangle = \langle M_b \cup X_i \cup Z_i\rangle$ and so by Eq. \ref{eq:errorgroup}, $\Pi_{M_b}$ and $\Pi_{M_b} \circ \Pi_{X_i} \circ \Pi_{Z_i}$ implement the same channel. Note that qubit $i$ of the input state has now become maximally mixed. This motivates us to define depolarized qubits as follows.

\begin{definition} \label{def:depolarization} [Depolarized Qubits]
    For any qubit $i$, if $X_i,Z_i \in \langle M_b\rangle$, then we say qubit $i$ is \textit{depolarized}. Otherwise, qubit $i$ is \textit{non-depolarized}. 
\end{definition}

Our algorithm will rely on the fact that with high probability an overwhelming majority of qubits become depolarized in this way. We will later show that in this case the depolarized qubits partition the circuit into small islands of non-depolarized qubits that can be simulated independently.

\subsection{Description of the Algorithm}
 
We can now write out the full algorithm, which is illustrated in Fig.~\ref{fig:diagram}. To summarize, we will sample some configuration of error locations, use this to form subsets of non-depolarized qubits that can be simulated independently, and apply a standard sampling-to-computing reduction (e.g.\cite{bremner_achieving_2017}) to sample using the fact that we can compute marginals efficiently. We prove correctness below in \cref{lemma:correctness}, and leave the full proof of runtime to \cref{sec:runtime}

\begin{samepage}
    \begin{algorithm}[]

\SetKwInOut{Promise}{Promise}
\SetKwInOut{Parameter}{Parameter}
\caption{Sampler for noisy Clifford-magic circuits}
\KwInput{Depth-$d$ circuit $C$, $n$-qubit state $\rho$, POVM $\{F_x\}_x$, error rate $\gamma$}
\Promise{$C$ is a Clifford circuit, $\rho$ is a product-state, $F_x = \otimes_i \mc U_i(\ketbra{x}{x})$ where each $\mc U_i$ is an abritrary single-qubit gate acting on qubit $i$}
\KwOutput{A sample from the probability distribution $p_{\tilde{\mc C}}$ defined by $p_{\tilde{\mc C}}(x) = \Tr(F_x \Tilde{\mc C}(\rho))$}
 Let $M_b = \{\}$ \label{step:1}
 \label{alg:1}
 
 For each depolarizing channel acting on qubit $i$ after timestep $t$, with probability $\gamma$ add both $\mc C_t^\dagger(X_i)$ and $\mc C_t^\dagger(Z_i)$ to $M_b$. \label{step:2}
 
 For each qubit $i$, check if both $X_i,Z_i \in \langle M_b \rangle$. If so, label this qubit as `depolarized' and as `non-depolarized' otherwise. \label{step:3}

 Compute a generating set $G$ for $C (\langle M_b \rangle)$ by finding a basis for the nullspace of $\mathrm T_{M_b}\Lambda$ over $\mathbb Z_2$. 
 \label{step:4}
 
 For each non-depolarized qubit $i$, enumerate the set of qubits in its forward lightcone. Merge together all sets that intersect to form new nonintersecting sets $\{L_j\}_j$. \label{step:5}
 
 For each $L_j$, let $G_j \subseteq \hat{\mathsf{P}}_{|L_j|}$ be the set of Pauli operators defined only on $L_j$, which is formed by taking each element of $G$ and truncating away all qubits outside of $L_j$. \label{step:6}

  For each $L_j$, store the final state in the Pauli basis as $\Tr_{-L_j}(\mc C \circ \Pi_{M_b}( \rho)) = \frac{1}{2^{|L_j|}} \sum_{s \in \langle G_j \rangle} \Tr(\rho s) \mc C(s)$ and use this to compute any marginal of the output distribution on $L_j$ \footnotemark. Sample measurement outcomes by a standard sampling-to-computing reduction \cite{bremner_achieving_2017} \label{step:7}.

 For all qubits $x \not\in L_j$ for all $j$, output a uniformly random bit as their measurement outcome. \label{step:8}
 
\end{algorithm}
\end{samepage}
\footnotetext{the explicit expression for each marginal is in \cref{eq:outputprob}}
For convenience, we will call each set of qubits denoted by $L_j$ above as a \textit{non-depolarized island}.

\begin{lemma} \label{lemma:correctness} (Correctness of Algorithm \ref{alg:1})
    For $\tilde{\mc C}$ defined as in \cref{theorem:main}, let $q_{\tilde{\mc C}}(x)$ be the distribution over output bitstrings $x$ produced by Algorithm \ref{alg:1} on input state $\rho$. Then, $p_{\tilde{\mc C}} = q_{\tilde{\mc C}}$.
\end{lemma}

\begin{proof}
    
First recall that by \cref{eq:xiziimpliesdepol}, each depolarized qubit can be replaced with the maximally mixed state $\eye/2$ at the beginning of the circuit. It then follows that any error or gate acting solely on maximally mixed qubits can be ignored since $\mc E(\eye/2) = \eye/2$ and $\mc U (\frac{\eye}{2} \otimes \frac{\eye}{2}) = \frac{\eye}{2} \otimes \frac{\eye}{2}$. Removing these from the circuit only leaves the errors and gates that are within the lightcone of a non-depolarized qubit. In other words, each remaining error and gate is contained within one of the sets in $\{L_j\}_j$ which are each disjoint. Crucially, this implies that the output state is unentangled across these sets of qubits and so each set can be simulated independently. Furthermore the qubits outside of any of these sets remain maximally mixed and so their measurement outcomes can be exactly simulated by uniformly random bits.

To simulate $\Tr_{-L_j}\mc C \circ \Pi_{M_b} (\rho)$, it is helpful to view the initial state in the Pauli basis as $\rho=\frac{1}{2^{n}} \sum_{s \in \hat{\mathsf{P}}_{n}} \Tr(\rho s) s$. Next, it can be shown that the only Pauli operators that pass through $\Tr_{-L_j}\mc C \circ \Pi_{M_b}$ are those in $\langle G_j \rangle$. This can be seen by first observing that only Pauli operators in $\mathsf{C}(\langle M_b \rangle)$ can pass through $\Pi_{M_b}$ by \cref{lemma:survivingpaulis}. This notably only includes Pauli operators with support on non-depolarized qubits. Additionally, any Pauli operator with support outside of $L_j$ will not pass through $\Tr_{-L_j} \mc C \circ \Pi_{M_b}$ since the lightcone of a non-depolarized qubit outside of $L_j$ is disjoint from $L_j$. Now, note that $\langle G_j \rangle$ clearly contains all elements in $\mathsf{C}(\langle M_b \rangle) \rangle$ that only have support in $L_j$, but it is not obvious that it does not also contain elements that are outside of $ \mathsf{C}(\langle M_b \rangle)$. It can be shown that this is also true by noting that each element of $M_b$ cannot have support on two non-depolarized qubits with non-intersecting lightcones. Therefore, each element in $\mathsf{C}(\langle M_b \rangle)$ will still commute with each element in $M_b$ even after replacing qubits outside of $L_j$ with identity. Thus, $\langle G_j \rangle$ exactly represents the Pauli operators that can pass through  $\Tr_{-L_j}\mc C \circ \Pi_{M_b}$ giving the equality: $\Tr_{-L_j}(\mc C \circ \Pi_{M}( \rho)) = \frac{1}{2^{|L_j|}} \sum_{s \in \langle G_j \rangle} \Tr(\rho s) \mc C(s)$.

We can now express the output probability of outcome $z$ as
\begin{align}
\label{eq:outputprob}
    \Tr(\otimes_{i\in L_j}\mc U_i(\ketbra{z}{z}) \Tr_{-L_j}\mc C \circ \Pi_{M_b}(\rho)) = \frac{1}{2^{|L_j|}} \sum_{s \in \langle G_j \rangle} \Tr(\rho s) \Tr(\otimes_{i\in L_j}\mc U_i(\ketbra{z}{z})\mc C(s))
\end{align}
\end{proof}

\subsection{Bounding the Number of Surviving Pauli Operators}
The primary runtime cost is due to Step \ref{step:7}, which we will show has a runtime that scales as $O(d^{3+D}n^4)$. This asymptotically dominates the runtime for finding a basis for the nullspace of a binary matrix over $\mathbb Z_2$ which can be done in $O(n^3)$ timesteps. 

The runtime to simulate qubits in the set $L_j$ is $O(d|L_j|^2|\langle G_j \rangle |)$. This is because it takes $O(d|L_j||\langle G_j \rangle |)$ timesteps to first evolve each Pauli operator through the Clifford circuit since there are $d|L_j|$ gates to apply and $|\langle G_j \rangle |$ operators to apply them on. Next, it takes $O(|L_j||\langle G_j \rangle |)$  to compute each output probability since there are $|\langle G_j \rangle |$ terms to sum in \cref{eq:outputprob} and $|L_j|$ trace products to compute for each term -- one for each qubit. Finally we must compute $|L_j|$ output probabilities for the sampling-to-computing reduction giving a total runtime of $O(d|L_j||\langle G_j \rangle |+ |L_j|^2|\langle G_j \rangle |) = O(d|L_j|^2|\langle G_j \rangle |)$. 

Next we bound $\E |\langle G_j \rangle |$. In the noiseless case (i.e. when $M_b$ is empty), the best bound we can get is $|\langle G_j\rangle | \leq 4^{|L_j|}$, since there are four possible single-qubit Pauli operators in each location. However, noise causes the expected number of `surviving' Pauli operators to decay with depth and noise strength. This is bounded in the following lemma, which is proven in \cref{appendix:counting}, and is obtained due to the \textit{unitarity} of the circuit.

\begin{lemma}  \label{lemma:num_surviving_paulis}
    Using the notation of Algorithm \ref{alg:1},
    \begin{align}
        \E|\langle G_j \rangle | \leq d e^{3(1-\gamma)^d |L_j|}
    \end{align}
\end{lemma}
\begin{corollary}\label{corollary:runtime_component}
    Suppose $T_j$ is the runtime required to sample from $\Tr_{-L_j}(\mc C \circ \Pi_{M_b}(\rho))$. Then, $\E[T_j] \leq cd^2 |L_j|^2 e^{3(1-\gamma)^d |L_j|}$ where $c$ is some constant representing the cost of applying Clifford unitaries on Pauli operators.
\end{corollary}

In the case that there is no geometric restriction on the Clifford circuit, we obtain our result by assuming there is only one connected component $L_1$ which contains all qubits and applying the above bound for $d < O(\gamma^{-1} \log n)$. In order to improve this depth threshold to constant for the geometrically local case we also need to show that each $L_j$ is of size logarithmic in $n$ by taking advantage of percolation.

\subsection{Percolation into Small Components}

To simplify the analysis, we coarse-grain our lattice into sublattices of length $2d$ in each dimension where $d$ is the circuit depth and define a connectivity graph $G = (V,E)$ where each vertex represents a sublattice and each edge represents whether two sublattices are adjacent (including diagonally adjacent). Importantly, two qubits can have intersecting forward lightcones if and only if they are in adjacent sublattices because a lightcone can spread at most a Manhattan distance of $d$ and so it can never reach the lightcone of a qubit in a non-adjacent sublattice. Hereafter, we define a non-depolarized sublattice to mean a sublattice that contains at least one non-depolarized qubit. Our main goal is to show that there does not exist any large connected components of non-depolarized sublattices. 

To prove this, we show that if there exists a large connected component of non-depolarized sublattices then there exists a high-weight Pauli operator on these sublattices that `survives' the Pauli projection channels.
\begin{lemma}\label{lemma:existshighwtpauli}
     For any error configuration $b$, let $A \subseteq V$ be the largest connected component in $G$ such that each sublattice of $A$ has at least one non-depolarized qubit. There exists some $s \in \mathsf{C}(\langle M_b \rangle)$ with support in at least half of the sublattices in $A$ and no support outside of $A$.
\end{lemma}

By the contrapositive of \cref{lemma:existshighwtpauli}, if there is no Pauli operator with high weight on a large connected component of sublattices that survives the error channel, then there are no large connected components of sublattices that have non-depolarized qubits.
We show in \cref{lemma:no_large_component_paulis} of \cref{appendix:runtime} that when the depth exceeds some constant depth threshold $d_{local}^*$, it is possible to tightly bound the probability that there exists some Pauli operator with high-weight support on a connected component of size $k$, using proof techniques from percolation theory. Combining this with \cref{lemma:existshighwtpauli} gives us the following bounds on the size of connected components,
\begin{corollary} \label{corollary:prob_x}
Let $\{L_j\}_j$ be the partition of qubits into subsets in Alg. \ref{alg:1}.
\begin{align*}
			\mathbf{P}(\max_j |L_j| \geq x) &\leq nd \exp(-x/(2d)^D) \\
             \mathbf{P}(\max_j diam(L_j) \geq x) &\leq 2^{3^D}nd \exp(-\Omega(\frac{\gamma x}{3^DD}))
		\end{align*}
when $d > d_{local}^*$, where  $d_{local}^* = O(\gamma^{-1}3^{2D}+ \gamma^{-1}D\log(\gamma^{-1}D))$. $diam(L_j)$ denotes the maximum manhattan distance between two qubits in $L_j$.
\end{corollary}

\subsection{Computing the Runtime} \label{sec:runtime}

Now, we have the ingredients to prove the main theorem.

\begin{proof} (Proof of \cref{theorem:main} and \cref{corollary:monte-carlo})
    Without geometric locality, we assume there is one connected component $L_1$ with all qubits. By examining \cref{corollary:runtime_component}, one can see that when $d > O(\gamma^{-1}\log n)$, the expected runtime is polynomial in $n$. This proves the first part of \cref{theorem:main}.
    For the geometrically local case, we `split' the noise channels into two parts of strength roughly $~\gamma/2$. More specifically, define $\gamma' = 1-\sqrt{1-\gamma}$, and replace each noise depolarizing noise channel of strength $\gamma$ with two consecutive independent depolarizing channels of strength $\gamma'$. We use one set of these noise channels to induce percolation into small components, and another set of these noise channels to decay the number of surviving Pauli operators within each component. When $d > d_{local}^*$ (where $d_{local}^*$ is defined for $\gamma'$ rather than $\gamma$ \footnote{Abusing notation, we refer to this threshold as $d_{local}^*$ as well since it has the same asymptotic scaling.}), we can compute the expected runtime $T$ of our algorithm as a whole, using linearity of expectation, as follows.
    \begin{align}
        \E[T] &\leq \sum_{j} \E [T_j]\\
        &\leq  \sum_{j} \sum_{x = 0}^{n}  c d^2 x^2 e^{3(1-\gamma')^d x} \mathbf{P}(|L_j| = x) \\
        &\leq  n\sum_{x = 0}^{n}  c d^2 x^2 e^{3(1-\gamma')^d x} nd e^{-\frac{x}{(2d)^D}} \\
        &\leq c d^3n^4 \sum_{x = 0}^{n}  e^{\frac{x}{(2d)^D}(3(2d)^D(1-\gamma')^d-1)}\\
        &\leq c  d^3n^4\sum_{x = 0}^{n}  e^{-\frac{2x}{3(2d)^D}} \label{eq:use_lambert_bound}\\
        &\leq c d^3n^4 \frac{1}{1-e^{-\frac{2}{3(2d)^D}}} \\
        &= O( d^{3+D} n^4) 
    \end{align}
    where we have obtained \cref{eq:use_lambert_bound} by noting that $3(2d)^D(1-\gamma')^d < 1/3$ when $d > d_{local}^*$ and $D \geq 1$ using \cref{eq:lambert_bound} proved in \cref{appendix:runtime}. The second to last inequality uses the geometric sum equation, and the final inequality uses the fact that $e^{-a} \leq 1-\frac{a}{2}$ for $a \leq 1/2$. \cref{lemma:correctness} tells us that the sampling algorithm is exact, and together with our bounds on expected runtime, this completes the proof of the theorem. 

    We now prove \cref{corollary:monte-carlo}, which is the version of the algorithm more likely to be implemented in practice. By Markov's inequality, the runtime will exceed $\epsilon^{-1}E[T]$ with probability $< \epsilon$ over the possible error configurations. We can check whether such an adversarial error configuration is sampled at Step \ref{step:7} of Algorithm \ref{alg:1} (i.e. if $|\langle G_j \rangle|$ is too large). If this occurs, we can stop the algorithm and output a random bit string. The TVD error incurred in this case is at most $1$, which means the expected TVD error (when the expectation is taken over possible error configurations) is bounded by $\epsilon$. By standard arguments, this means the resultant distribution $q_{\tilde{\mc C}}$ $\epsilon$-approximates $p_{\tilde{\mc C}}$ in total variation distance \footnote{See \cite{rajakumar_polynomial-time_2024}, corollary 23, for a more explicit proof}.
\end{proof}

\section{Relation to Existing Circuit Classes}\label{sec:relation}

\subsection{Clifford-Magic and Conjugated Clifford Circuits}
\begin{definition}[Clifford-magic (CM) circuits \cite{yoganathan_quantum_2019}]
    Clifford-magic circuits are composed of Clifford gates applied to magic-state input. More specifically, the circuit can be denoted as $C = U_d \dots U_1$ where $U_i \in \mathsf{C}_k$. And the input state is $\rho = \ketbra{A}^{\otimes n}$ where $\ket{A} = \frac{1}{\sqrt{2}}(\ket{0} + e^{i\pi/4}\ket{1})$. The output is measured in the computational basis.
\end{definition}

\begin{definition}[Conjugated Clifford circuits (CCC) \cite{bouland}]
    A conjugated Clifford gate is a Clifford conjugated with an arbitrary single-qubit rotation gate $U$ on every qubit. Conjugated Clifford circuits (CCCs) are composed of such gates, where the input state is $\ketbra{0}{0}^{\otimes n}$ and the output is measured in the computational basis.
\end{definition}

Noisy Clifford-magic circuits trivially fall into the model of circuits we simulate. To see how noisy Conjugated Clifford circuits do as well, conjugate each depolarizing channel with $U$ (which leaves it unchanged). Then, each $U$ and $U^\dagger$ between gates and noise channels in the circuits can be cancelled out leaving a circuit of the form $\mc U^{\otimes n} \circ \tilde{\mc C} \circ \mc U^{\dagger \otimes n}$, where $\tilde{\mc C}$ is some noisy Clifford circuit which can be simulated under our model.

\subsection{IQP Circuits augmented with CNOTs}

With a few modifications, our techniques also work for simulating noisy IQP circuits with intermediate CNOT gates. Like the circuit classes in the previous section, these have also been proposed for near-term quantum advantage experiments \cite{hangleiter_fault-tolerant_2024,bluvstein_logical_2024}. We define this circuit family more formally below:

\begin{definition}[Instantaneous Quantum Polynomial Circuits augmented with CNOTs (IQP+CNOT) \cite{hangleiter_fault-tolerant_2024,bluvstein_logical_2024}]
    An IQP+CNOT circuit is a circuit composed of any gates that are diagonal in the computational basis along with CNOT gates, where state preparation and measurement is done in the Hadamard basis. 
\end{definition}

In the case of IQP+CNOT circuits, our results apply to a broader class of noise channels, which notably include depolarizing and dephasing channels \footnote{This is the same class of noise channels for which noisy IQP circuits without CNOT gates become classical simulatable due to \cite{rajakumar_polynomial-time_2024}}. We state this in the following theorem. 
\begin{theorem}\label{theorem:iqp+cnot}
The results of \cref{theorem:main} apply in the case that $\rho = \ketbra{+}^{\otimes n}$ and $\tilde{C}$ is composed of only diagonal gates and CNOT gates. Moreover, $\mc E$ can be any noise channel in the following class of single-qubit Pauli noise channels,
\begin{align}
\mc E(\rho) = (1-p_X -p_Y-p_Z) + p_X X \rho X + p_Y Y \rho Y + p_Z Z \rho Z
\end{align}
where $\gamma = p_Z + \min(p_Y,p_Z)$.
\end{theorem}

Once again our strategy is to simulate mid-circuit errors by propagating them to the beginning of the circuit, but in this case we only propagate the $Z$ component of the error. These remain as $Z$ errors on the input since they commute with diagonal gates and are mapped to other $Z$ errors by the CNOT gates. Since $X$ errors do not commute with diagonal gates, we will not propagate them to the beginning, and instead stochastically simulate them mid circuit.
Notice that the projection channel resulting from these propagated errors is exactly the same as if we had replaced each IQP gate with identity and propagated both $X$ and $Z$ errors to the beginning as done previously. In this modified circuit, $Z$ errors would propagate in the same way as the original IQP+CNOT circuit and $X$ errors would propagate to other $X$ errors which would then act trivially on the input since $\ketbra{+}^n$ is a $+1$ eigenstate of any $s \in \mathsf{P}^X_n$. 
Since replacing diagonal gates with identity makes the circuit entirely Clifford, all of our previous analysis can be applied to show that many of the qubits become depolarized by this propagated error channel acting on the input. Intuitively, percolation of depolarizing errors still occurs even when only considering $Z$ errors because only a single $Z$ error is needed to depolarize an input qubit (i.e. $\Pi_{Z}(\ketbra{+}) = \frac{\eye}{2}$). We leave the full proof to \cref{appendix:iqpcnot}.


\section*{Appendix}

\appendix

\section{Counting Arguments for Pauli Paths in Clifford Circuits} \label{appendix:counting}

Note that since $\{X,Y,Z\}$ have trace zero, for any $s \in \mathsf{P}_n$, we have the property that $\mc E^{\otimes n}(s) = (1-\gamma)^{|s|}s$ where $|s|$ is the number of non-identities in $s$. This convenient property motivates studying the evolution under the noisy circuit in the Pauli basis. Similar techniques were developed and used in a large body of existing literature (\cite{bremner_achieving_2017,gao_efficient_2018,aharonov_polynomial-time_2023,schuster_polynomial-time_2024,fontana2023classicalsimulationsnoisyvariational,mele2024noiseinducedshallowcircuitsabsence} are examples and maybe an incomplete list). We refer to this broad class of techniques as the Pauli path framework. One of our technical contributions is to reframe these techniques under the application of stochastic processes rather than under the application of a deterministic CPTP map (in earlier work, the probability bound of \cref{eq:simple_survival_prob} is treated as a decay factor for the associated Pauli path coefficient). In the setting of Clifford circuits, the analysis can be greatly simplified, and for any $s \in \mathsf{P}_n$, we have the following fact,
\begin{align} \label{eq:simple_survival_prob}
         \mathbf{P}_b (\tilde{\mc C}_b (s) \neq 0) = (1-\gamma)^{\sum_i | \mc C_i(s)|} 
\end{align}
This can be seen by applying $\mc E^{\otimes n}$ on $\mc C_i(s)$ at each layer.

One fact that is unique to Clifford circuits is that any Pauli operator $s\in \mathsf{P}_n$ cannot `split' into different Pauli paths during the circuit (as is the case with general quantum circuits, e.g. random quantum circuits). It will be useful to examine the point in the circuit at which an input Pauli operator is at its minimum weight during its path through the Clifford circuit. We define this below,
\begin{definition}
    Given any set of qubits $A$, let $S_w^A(\mc C, G) \subseteq \hat{\mathsf{P}}_n$ be the set $S_w^A(\mc C, G):= \{s \in G : \min_i \mc C_i(s) = w \text{, and } \forall_i \text{ } \mc C_i(s) \text{ is supported only in $A$}\}$. When $A$ is not defined, we take it to be the set of all qubits. When $w$ is not defined, we take it to include all $w: 0 \leq w \leq |A|$. When $\mc C$ is not defined we take it to be the entire circuit, which can be inferred from context. When $G$ is not defined we take it to be $\hat{\mathsf{P}}_n$.
\end{definition}
By applying \cref{eq:simple_survival_prob}, we have that,
\begin{fact} \label{fact:survivalprob}
    For $s \in S_w^A$ and noisy Clifford circuit $\tilde{\mc C}$ of depth $d$
    \begin{align} 
    \mathbf{P}_b (\tilde{\mc C}_b(s) \neq 0) = (1-\gamma)^{\sum_i |\mc C_i(s)|} \leq (1-\gamma)^{dw}
\end{align}
\end{fact}
Crucially, this probability decays exponentially with depth and weight. Now, we can also bound the size of $S_w$ as follows
\begin{lemma} \label{lemma:size_of_sw}
    \begin{align}
        |S_w^A| \leq d {|A| \choose w}3^w
    \end{align}
\end{lemma}
\begin{proof}
    Note, for each $s \in S_w^A$, there exists some $i$ at which $|\mc C_i(s)|=w$ and its support is contained within $A$. Thus, for all $s \in S_w^A$,  the set of all possible weight-$w$ Pauli operators that are contained within $A$ contains $\mc C_i(s)$ for some $i$. By \textit{unitarity}, each such $\mc C_i(s)$ corresponds to exactly one input Pauli $s$. Therefore we can upper bound $|S_w^A|$ by counting all weight-$w$ Pauli operators contained in $A$ for each depth. We can ignore phases when performing this counting since $S_w^A \subseteq \hat{\mathsf{P}}_n$ and for each $P \in \hat{\mathsf{P}}_n$ there is exactly one phase $\ell \in \{0,1,2,3\}$ such that $\mc C_i^\dagger(i^\ell P) \in \hat{\mathsf{P}}_n$. The counting argument proceeds as follows. First, there are $d$ ways to choose the timestep $i$ at which $|\mc C(i)| = w$. Next, there are ${|A| \choose w}$ ways to choose the locations of $w$ non-identities, and $3^w$ ways to assign those locations to single-qubit Pauli operators (X,Y, or Z).
\end{proof}

This allows us to obtain the following bound:

\begin{lemma} (Restatement of \cref{lemma:num_surviving_paulis})
    Using the notation of Algorithm \ref{alg:1},  we have
    \begin{align}
        \E |\langle G_j \rangle | \leq d e^{3(1-\gamma)^d |L_j|}
    \end{align}
\end{lemma}

\begin{proof}
Note that once we have sampled some $b$, $\langle G_j \rangle = S^{L_j} \cap \mathsf{C}(\langle M_b \rangle)$. Moreover, this is equivalent to the set $\{s \in S^{L_j} : \tilde{\mc C}_b (s) \neq 0\}$. Now, applying linearity of expectation, we have the following,
    \begin{align}
        \E |\langle G_j \rangle | &= \E_b |\{s \in S^{L_j} : \tilde{\mc C}_b (s) \neq 0\}| \\
        &\leq \sum_{s \in S^{L_j}} \mathbf{P}_b( \tilde{\mc C}_b(s) \neq 0) \\
        &\leq \sum_{w=0}^{|L_j|}\sum_{s \in S^{L_j}_w} \mathbf{P}_b( \tilde{\mc C}_b (s) \neq 0) \\ 
        &\leq \sum_{w=0}^{|L_j|} d{|L_j| \choose w}3^w (1-\gamma)^{wd} &\text{by \cref{fact:survivalprob} and \cref{lemma:size_of_sw}} \\
        &\leq d (1+3(1-\gamma)^{d})^{|L_j|} &\text{by binomial theorem}\\
        &\leq d \exp(3(1-\gamma)^d |L_j|)
    \end{align}
\end{proof}

\section{Depth Threshold for Percolation} \label{appendix:runtime}

A useful fact that we will soon make use of is that for a Pauli subgroup $G \leqslant \hat{\mathsf{P}}_n$ 
\begin{align}
\label{eq:ccg}
    \mathsf{C}(\mathsf{C}(G)) = G
\end{align}

Although, we do not include a proof here, this can be seen by mapping from phaseless Pauli subgroups to symplectic subspaces where the equivalent statement is shown in Theorem 11.8 of \cite{Roman2008}.

\begin{lemma} (Restatement of \cref{lemma:existshighwtpauli})
For any error configuration $b$, let $A \subseteq V$ be the largest connected component in $G$ such that each sublattice of $A$ has at least one non-depolarized qubit. There exists some $s \in \mathsf{C}(\langle M_b \rangle)$ with support in at least half of the sublattices in $A$ and no support outside of $A$.
\end{lemma}

\begin{proof}
    We start by showing that for each non-depolarized qubit $i$ in each sublattice of $A$ there is a Pauli operator $s_i$ such that 
    \begin{enumerate}
        \item $\Pi_{M_b}(s_i) \neq 0$
        \item $s_i$ has support on qubit $i$
        \item Its support is contained within $A$
    \end{enumerate}
    Note that by \cref{lemma:survivingpaulis}, (1) is equivalent to $s_i \in \mathsf{C}(\langle M_b \rangle)$.
    
    We first show there exists an $s_i$ with properties (1) and (2). For sake of contradiction assume there does not a exist an $s_i$ such that (1) and (2) are true. Next, add $X_i$ and $Z_i$ to $M_b$ to produce $M_b'$. Since every element of $\mathsf{C}(\langle M_b\rangle)$ has identity on qubit $i$ by assumption, they each will commute with $X_i$ and $Z_i$, and so $\mathsf{C}(\langle M_b \rangle)$ = $\mathsf{C}(\langle M_b' \rangle)$. This implies $\langle M_b \rangle = \langle M_b'\rangle$ (by \cref{eq:ccg}), which implies that $X_i,Z_i \in M_b$. By definition this means that qubit $i$ is depolarized which is a contradiction. Thus, (1) and (2) must be true. 
    
    Next, we show that there exists an $s_i$ that additionally has property (3). To do this we start with an $s_i'$ with properties (1) and (2), and construct $s_i$ by setting all qubits of $s_i'$ that are outside of $A$ to identity. $s_i$ now satisfies (2) and (3). We next show that it also satisfies (1). 
    Let $N(A)$ denote the set of sublattices that are adjacent to $A$ and including $A$. First, notice that $s_i'$ and $s_i$ are both not supported in adjacent sublattices of $A$ as these qubits are depolarized (otherwise they would be included in $A$) and so no Pauli with non-identity on $N(A)-A$ can pass through $\Pi_{M_b}$. Now, choose any $P \in M_b$. If $P$ has support on $A$, then it cannot have support outside of $N(A)$ because each $P$ is contained within a lightcone of the circuit and lightcones can only spread to adjacent sublattices. Therefore, its commutation relation with $s_i$ will be the same as its commutation relation with $s_i'$ (as they are identical on $N(A)$). Alternatively, if $P$ does not have support in $A$, then they must commute as well since $s_i$'s support is contained entirely in $A$.

    Equipped with this set of Pauli operators $\{s_i\}_{i \in A}$ with the above three properties, we finish the proof using the probabilistic method. Define a randomly chosen Pauli operator $s$ as follows,
    \begin{align}
        s = \prod_{i \in A} s_i^{r_i}
    \end{align}
    where each $r_i$ is drawn from an independent random variable $R_i$ taking values $0$ or $1$ with equal probability. Notice that $s$ maintains properties (1) and (3). For each qubit $i \in A$ define the following event:

    \begin{align}
        E_i = \begin{cases}
            1 & \text{if } s \text{ has a non-identity on qubit $i$}\\
            0 & \text{otherwise}
        \end{cases}
    \end{align}
    Observe that $\E E_i \geq 1/2$. This can be seen by the following argument. Assume without loss of generality that $s_i$ contains an $X$ on qubit $i$. The argument works the same if it contains a $Y$ or $Z$ instead and we are guaranteed one of the three is true by property (2). Then let $S_X = \{j \in A: \text{$s_j$ contains an $X$ on qubit $i$}\}$. Notice that if $\oplus_{j \in S_X} r_j=1$ then $s$ has non-identity on qubit $i$. $\Pr(\oplus_{j \in S_X} R_j=1)=1/2$ as long as $S_X$ is non-empty, which is true by assumption, and so $\E E_i \geq 1/2$.
    Using linearity of expectation, we have,
    \begin{align}
        \E \sum_{i\in A} E_i = \sum_{i \in A} \E  E_i \geq |A|/2 
    \end{align}
    There must be at least one assignment to the random variables $\{R_i\}$ that achieves at least this expectation. Therefore, there exists some Pauli operator $s \in \mathsf{C}(\langle M_b\rangle)$ with support on at least $|A|/2$ of the sublattices and no support outside of $A$.
\end{proof}

Next, we must show that with high probability there does not exist any such Pauli operator that has support on half of the sublattices of a large connected component. This set of Pauli operators is formally defined below.

\begin{definition}
    Let $T_k \subseteq \hat{\mathsf{P}}_n$ be the subset of Pauli operators such that $s \in T_k$ if $s$ is contained within a connected component of $k$ sublattices and has support on at least $k/2$ sublattices within this connected component.
\end{definition}

Next, using a similar argument to \cref{appendix:counting}, we bound the number of Pauli operators in $T_k$ that have a minimum weight of $w$ in the circuit.
\begin{lemma}
\label{lemma:skwsize}
\begin{align}
    |T_k \cap S_w| \leq nd 2^{3^D k}{k (6d)^D \choose w} 3^w
\end{align}
\end{lemma}
\begin{proof}
    Recall that by definition each $s \in T_k$ is contained within a connected component of $k$ sublattices. There are at most $n2^{\Delta k}$ ways to choose a $k$-sized connected component (in which $s$ has support) in a $\Delta$-regular graph \cite{krivelevich2015phase}. Here $\Delta=3^D-1$ although we drop the minus one for convenience. When the circuit is applied, the support of $\mc C_i(s)$ is contained within the neighborhood of the original connected component due to lightcones. This consists of $3^Dk$ sublattices since each sublattice has $3^D$ neighbors including itself. We label the set of these qubits as $A$, where $|A| = 3^D k(2d)^D = (6d)^Dk$. Since $T_k \cap S_w\subseteq S^A_w$, it remains to upper bound $|S^A_w|$. By \cref{lemma:size_of_sw}, $|S^A_w| \leq d{k (6d)^D \choose w} 3^w$.
\end{proof}
Applying this upper bound on $|T_k \cap S_w|$, we can now bound the probability that any $s \in T_k$ survives $\Tilde{\mc C}$.

\begin{lemma}\label{lemma:no_large_component_paulis}
\begin{align}
    \mathbf{P}_b(\bigcup_{s \in T_k} \Pi_{M_b}(s) \neq 0) \leq nd\exp(-k)
\end{align}
when $d > d_{local}^*$, where  $d_{local}^* = O(\gamma^{-1}3^{2D}+ \gamma^{-1}D\log(\gamma^{-1}D))$
\end{lemma}
\begin{proof}
First note that for any timestep $i$ and any $s \in T_k$, $|\mc C_i(s)| \geq \frac{k}{2\cdot 3^D}$. To show this, we apply a similar lightcone argument to the proof of \cref{lemma:skwsize}. If $\mc C_i(s)$ has support on $x$ sublattices, then $\mc C_i^\dagger (\mc C_i(s)) = s$ has support on at most $3^Dx$ sublattices due to lightcones. Therefore, since $s$ has support on $k/2$ sublattices, this guarantees that $\mc C_i(s)$ has support on at least $k/(2 \cdot 3^D)$ sublattices.

\begin{align}
     \mathbf{P}_b(\bigcup_{s \in T_k} \Pi_{M_b}(s) \neq 0)  &\leq \sum_{s \in T_k}  \mathbf{P}_b(\Pi_{M_b}(s) \neq 0) \\
    \label{eq:step2}
    &= \sum_{w=k/(2 \cdot 3^D)}^{k (6d)^D } \sum_{s \in T_k \cap S_w} (1-\gamma)^{\sum_i | \mc C_i(s)|}  \\
    \label{eq:step3}
    & \leq \sum_{w=k/(2 \cdot 3^D)}^{k (6d)^D } |T_k \cap S_w| (1-\gamma)^{wd} \\
    \label{eq:step4}
    & \leq (1-\gamma)^{kd/(4 \cdot 3^D)} \sum_{w=1}^{k (6d)^D } |T_k \cap S_w| (1-\gamma)^{wd/2} \\ 
    \label{eq:step5}
    & \leq (1-\gamma)^{kd/(4 \cdot 3^D)}\sum_{w=1}^{k (6d)^D }  nd 2^{3^D k} {k (6d)^D \choose w} 3^w (1-\gamma)^{wd/2}  \\
    \label{eq:bin}
    &= (1-\gamma)^{kd/(4 \cdot 3^D)} nd2^{3^Dk} (1+3(1-\gamma)^{d/2})^{k (6d)^D} \\
    &\leq nd\exp( \frac{k(-\gamma d)}{4 \cdot 3^D}+k3^D\ln 2+ 3(1-\gamma)^{d/2}k(6d)^D)\\
    &\leq nd\exp(- k [\frac{\gamma d}{4 \cdot 3^D} - 3^D \ln 2 - 3(1-\gamma)^{d/2}(6d)^D ]) \label{eq:final_bound}
\end{align}
\cref{eq:step2} and \cref{eq:step3} follow from \cref{fact:survivalprob}, \cref{eq:step4} splits the error term into two factors, \cref{eq:step5} applies \cref{lemma:skwsize}, and \cref{eq:bin} applies the binomial theorem. We now want to show that the following term is $\geq 1$ in some depth regimes,
\begin{align}
    \frac{\gamma d}{4 \cdot 3^D} - 3^D \ln 2 - 3(1-\gamma)^{d/2}(6d)^D 
\end{align}
We can lower bound the first two terms $\frac{\gamma d}{4 \cdot 3^D} - 3^D \ln 2 \geq 2$ when $d \geq \gamma^{-1}(4 \cdot 3^D)(3^{D} \ln 2 + 2)$. We can write the third term as,
\begin{align} \label{eq:lambert_bound}
    3(1-\gamma)^{d/2}(6d)^D &\leq 3(6de^{-\frac{\gamma d}{2D}})^D \\
    &= 3(6d \frac{4D}{\gamma}\frac{\gamma}{4D} e^{-\frac{\gamma d}{2D}})^D\\
    &< 3(6\frac{4D}{\gamma}e^{-\frac{\gamma d}{4D}})^D &\text{using the fact that $xe^{-x} < 1$} \\
    &= 3(6\frac{4D}{\gamma})^D e^{-\frac{\gamma d}{4}}\\
    &\leq 1
\end{align}
where the last inequality holds true when $d \geq O(\gamma^{-1}D\log(\gamma^{-1}D))$. Thus, we have shown that $\frac{\gamma d}{4 \cdot 3^D} - 3^D \ln 2 - 3(1-\gamma)^{d/2}(6d)^D \geq 1$ when $d > d^*_{local}$ where $d^*_{local} = O(\gamma^{-1}3^{2D} + \gamma^{-1}D\log(\gamma^{-1}D))$.

\end{proof}

\begin{corollary} (Restatement of \cref{corollary:prob_x})
Let $\{L_j\}_j$ be the partition of qubits into subsets in Alg. \ref{alg:1}.
\begin{align*}
			\mathbf{P}(\max_j |L_j| \geq x) &\leq nd \exp(-x/(2d)^D) \\
             \mathbf{P}(\max_j diam(L_j) \geq x) &\leq 2^{3^D}nd \exp(-\Omega(\frac{\gamma x}{3^DD}))
		\end{align*}
when $d > d_{local}^*$, where  $d_{local}^* = O(\gamma^{-1}3^{2D}+ \gamma^{-1}D\log(\gamma^{-1}D))$. $diam(L_j)$ denotes the maximum manhattan distance between two qubits in $L_j$.
\end{corollary}
\begin{proof}
    Note that each set $L_j$ in Algorithm \ref{alg:1} spans at least $|L_j|/(2d)^D$ sublattices. Therefore by \cref{lemma:existshighwtpauli} there exists $s \in T_{k}$ for $k \geq \max_j|L_j|/(2d)^D$ such that $\Pi_{M_b}(s) \neq 0$. Using \cref{lemma:no_large_component_paulis}, this suffices as proof of the first inequality. When considering the diameter of $L_j$, we can make a few improvements. Firstly, when $d > d^*_{local}$,  \cref{lemma:no_large_component_paulis} can be tightened to 
        \begin{align}
            \mathbf{P}_b(\bigcup_{s \in T_k} \Pi_{M_b}(s) \neq 0) \leq 2^{3^D} nd\exp(-\Omega(\frac{\gamma d k}{3^D}))
        \end{align}
        by a simple examination of \cref{eq:final_bound}. Next, if $L_j$ contains two qubits of manhattan distance $x$, then it spans at least $k \geq x/2dD$ sublattices (regardless of $D$) by considering the case where they are diagonal. Combining these two facts gives the second inequality.
\end{proof}

\section{Subroutine for IQP + CNOT Circuits}
\label{appendix:iqpcnot}

\begin{theorem}(Restatement of \cref{theorem:iqp+cnot})
The results of \cref{theorem:main} apply in the case that $\rho = \ket{+}^{\otimes n}$ and $\tilde{C}$ is composed of diagonal gates and CNOT gates. Moreover, $\mc E$ can be any noise channel in the following class of single-qubit Pauli noise channels,
\begin{align}
\mc E(\rho) = (1-p_X -p_Y-p_Z) + p_X X \rho X + p_Y Y \rho Y + p_Z Z \rho Z
\end{align}
where $\gamma = p_Z + \min(p_Y,p_Z)$. Notably, this includes depolarizing and dephasing channels.
\end{theorem}
\begin{proof}
We make the following edits to the steps of Algorithm \ref{alg:1}. 
\begin{itemize}
    \item Step \ref{step:2}: Add only $\mc C^\dagger_t (Z_i)$ to $M_b$ instead of both $\mc C^\dagger_t (Z_i)$ and $\mc C^\dagger_t (X_i)$. Simulate the $\Pi_{X_i}$ in place by adding either an $X_i$ or an identity to the circuit with equal probability. As noted in \cite{rajakumar_polynomial-time_2024}, this technique allows us to generalize to dephasing noise channels and channels which can be written as a composition of dephasing channels with other Pauli noise channels.
    \item Step \ref{step:3}: We will say qubit $i$ is `depolarized' if $Z_i \in \langle M_b \rangle$ (since $\Pi_{Z}(\ketbra{+}) = \frac{\eye}{2}$).
    \item Step \ref{step:4}: Compute a generating set $G$ for $\mathsf{C}(\langle M_b \rangle ) \cap \mathsf{P}^X_n$ rather than the full group $\mathsf{C}(\langle M_b \rangle )$, since the Pauli basis decomposition of the input state $\ketbra{+}^n$ contains only operators in $\mathsf{P}^X_n$
\end{itemize}

After the above modifications, $M_b \subseteq \mathsf{P}^Z_n$, while $ G \subseteq \mathsf{P}^X_n$. Due to this fact, most of the analysis for the Clifford case imports directly over. The only step of the analysis that is not as trivial is computing the output probabilities for the distribution on $L_j$ after writing the output state in the Pauli basis. In the Clifford case, this is simple because one can efficiently evolve each Pauli basis element as $\mc C(s)$ for each $s \in \langle G_j \rangle$. In the IQP+CNOT setting, the non-Clifford diagonal gates may act non-trivially on the Pauli operators, so simulation in the Pauli basis may not be ideal. In other words, going from \cref{lemma:num_surviving_paulis} to \cref{corollary:runtime_component} is not simple. 

First we note that the input state can be written in the Pauli basis as:
\begin{align}
    \rho = \frac{1}{2^n} \sum_{s \in \langle G_j\rangle } \Tr (\ketbra{+}^ns) s =\frac{1}{2^n} \sum_{s \in \langle G_j\rangle } s
\end{align}
Next, \cref{lemma:iqp_converter} (proved below) tells us that states of this form also have a sparse representation in the computational basis. That is, 
\begin{align}
\rho = \E_{r \in \{0,1\}^{n-|G|}}\ketbra{\psi_r}
\end{align}
where each $\ket{\psi_b}$ is a superposition of only $2^{|G_j|}$ computational basis states. Once we have this representation, we can sample $r \in \{0,1\}^{n-|G_j|}$ randomly, and simulate $C \ket{\psi_r}$ in the computational basis by updating its state vector description after each CNOT and diagonal gate. Note that neither CNOTs nor diagonal gates can create a superposition of computational basis states, so we will only need to store $2^{|G_j|} = |\langle G_j \rangle |$ computational basis states throughout the entire circuit evolution. We can then sample single-qubit measurement outcomes on the final state using standard state vector methods.

\end{proof}

\begin{lemma} \label{lemma:iqp_converter}
    For any set of Pauli operators $G \leq \mathsf{P}^X_n$ such that $\rho = \frac{1}{2^n}\sum_{s \in \langle G\rangle } s$ is a valid quantum state, one can find a decomposition of $\rho$ of the following form
    \begin{align}
        \rho = \E_{r \in \{0,1\}^{n-|G|}}\ketbra{\psi_r}
    \end{align}
    where $\ket{\psi_b} = \sum_{i \in \{0,1\}^{|G|}}\frac{1}{\sqrt{2^{|G|}}} \ket{f(i,r)} $ for some efficiently computable function $f: \{0,1\}^n \to \{0,1\}^n$.
\end{lemma}
\begin{proof}
    Let $\mathrm T^X_G$ be the first $n$ columns of the matrix $\mathrm T_G$. In other words it is the $X$ component of the tableau matrix for $G$, which only contains Pauli operators in $\mathsf{P}^X_n$ anyways. It can be shown that there exists an $n \times n$ binary matrix $A$ such that

    \begin{align}
        T^X_G A = \begin{pmatrix}[cc]
                    \eye_{|G|} & 0
        \end{pmatrix}
    \end{align}
    where $A$ is composed of the elementary column operations associated with performing Gauss-Jordan elimination on the transpose of $\mathrm T^X_G$. The above form follows from the fact that $\mathrm T^X_G$ has rank $|G|$ since $G$ is a basis of independent operators.
    
    

     Note, elementary column operations on binary matrices are either (1) swapping column $i$ and column $j$ or (2) setting column $i$ to be column $i$ $\oplus$ column $j$. Since each column represents a qubit, operation (1) corresponds to a SWAP operation between qubit $i$ and $j$ while operation (2) corresponds to a CNOT operation targetted on qubit $i$ and controlled on qubit $j$. Therefore, $A$ can be associated with a unitary matrix $U_A$ which performs these operations.

     Notice now that,
     \begin{align}
         \mc U_A(\rho) &= \mc U_A \left( \frac{1}{2^n}\sum_{b_1,\ldots,b_{|G|} \in \{0,1\}}  \prod_{s \in G} s^{b_i} \right)\\
         &= \frac{1}{2^n} \sum_{b_1,\ldots,b_{|G|} \in \{0,1\}} \prod_{s \in G} \mc U_A(s^{b_i})\\
         &= \frac{1}{2^n} \sum_{b_1,\ldots,b_{|G|} \in \{0,1\}}\prod_{i \in |G|} X_{i}^{b_i}\\
         &= \left(\frac{\eye+X}{2} \right)^{\otimes |G|} \otimes \left(\frac{\eye}{2}\right)^{\otimes n-|G|}\\
         &=  \ketbra{+}^{\otimes |G|} \otimes \left(\frac{\eye}{2}\right)^{\otimes n-|G|}\\
         &=  \E_{r \in \{0,1\}^{n-|G|}} \ketbra{+}^{|G|} \otimes \ketbra{r}\\
     \end{align}
     Therefore, $\rho = \mc U_A^{\dagger} (\E_{r \in \{0,1\}^{n-|G|}} \ketbra{+}^{|G|} \otimes \ketbra{r})$. Since $U_A^{\dagger}$ is composed of SWAP and CNOT gates, it maps computational basis states to computational basis states. Defining $f : \{0,1\}^n \to \{0,1\}^n$ to be the function such that $f(x) =\mc U_A^{\dagger}(\ketbra{x})$ concludes the proof.
\end{proof}

\section{Noise-induced Anticoncentration in Clifford Circuits}
\label{app:anticoncentration}

While the circuits we consider are known to be hard to exactly sample from in the worst-case, an additional property known as anticoncentration is usually required in order to show hardness of approximate sampling. In particular, Clifford-Magic, Conjugated Clifford, and IQP+CNOT circuits were all initially proposed as an ensemble of random circuits which anticoncentrate \cite{bouland,yoganathan_quantum_2019,hangleiter_fault-tolerant_2024,bluvstein_logical_2024}. This anticoncentration property enables one to prove hardness of approximate sampling for these circuit families (up to complexity-theoretic assumptions), but at the same time, it allows for polynomial-time classical simulation when the circuit is noisy \cite{bremner_achieving_2017,aharonov_polynomial-time_2023,schuster_polynomial-time_2024}. Thus, an important line of work is to understand the regimes in which anticoncentration sets in. Here, we show that our counting arguments can be adapted to show that \emph{any} noisy Clifford circuit on random input bit-strings anticoncentrates in $O(\log n)$ depth. When considering an ensemble of circuits where Clifford gates are drawn randomly and states are prepared and measured in an arbitrary product basis, the randomness in the input bit-strings can be absorbed into the circuit. Finally, since random Clifford gates form a 2-design for Haar random gates \cite{Dankert_2009,DiVincenzo_2002,webb2016Cliffordgroupformsunitary,zhu_Clifford_2016}, this result also implies that noisy Haar random circuits anticoncentrate in $\Theta(\log n)$-depth for all architectures, due to noise. This was only previously known for $1D$ and all-to-all architectures. 

\begin{theorem}
    Using the same notation as \cref{theorem:main}, suppose we use $p_{\tilde{\mc C},y}$ to denote the probability distribution generated by $\tilde{\mc C}$ applied to some computational basis state $\ketbra{y}$ and measured in an arbitrary basis. Then, when $d \geq O(\gamma^{-1} \log n)$:
    \begin{align}
        \E_{y \sim \{0,1\}^n} \sum_{x\in\{0,1\}^n} p_{\tilde{\mc C},y}(x)^2 = \frac{O(1)}{2^n}
    \end{align}
where the expectation is taken over the uniform distribution.
\end{theorem}
\begin{proof}
    We will make use of the following orthogonality property. For any $s,s' \in \hat{\mathsf{P}}^Z_n$, if $s \neq s'$, then
    \begin{align}
    \E_{y \sim \{0,1\}^n} \Tr(s \ketbra{y})\Tr(s' \ketbra{y}) = 0
    \end{align}
    Letting $v_z(\cdot)$ be the map to the binary symplectic representation where we have dropped the X-component, the above property can be seen as follows:
    \begin{align}
        \E_{y \sim \{0,1\}^n} \Tr(s \ketbra{y})\Tr(s' \ketbra{y}) &= \E_{y \sim \{0,1\}^n} -1^{v_z(s)\cdot y + v_z(s')\cdot y} \\
        &= \E_{y \sim \{0,1\}^n}-1^{(v_z(s)+ v_z(s')) \cdot y} \\
        &= 0 &\text{if $v_z(s)+ v_z(s') \neq 0$}
    \end{align}
    We will relate anticoncentration to the expected number of Pauli operators which `survive' the errors. Recall that given a configuration of errors $b$ in circuit $\tilde{C}_b$, we denote the group of propagated errors as $M_b$ and the corresponding centralizer (`surviving' Pauli operators) as $\mathsf C(\langle M_b\rangle)$. Using $U$ to denote the unitary which transforms the measurement basis into the computational basis, we have,
    \begin{align}
        \E_{y \sim \{0,1\}^n} \sum_{x\in\{0,1\}^n} p_{\tilde{\mc C},y}(x)^2 &=\E_{y \sim \{0,1\}^n} \sum_{x\in\{0,1\}^n} \Tr(\ketbra{x} \mc U \circ \tilde{\mc C} (\ketbra{y}))^2 \\
        &= \E_{y \sim \{0,1\}^n} \sum_{x\in\{0,1\}^n} \Tr(\ketbra{x} \E_{b} \left[\mc U \circ \mc C \circ \Pi_{M_b} (\ketbra{y})\right])^2 \\
        &= \E_{y \sim \{0,1\}^n} \sum_{x\in\{0,1\}^n} \left(\frac{1}{2^n}\sum_{s \in \hat{\mathsf{P}}_n^Z}  \E_{b}\left[ \Tr( \ketbra{x}\mc U \circ \mc C \circ \Pi_{M_b} (s)) \Tr(s \ketbra{y})\right]\right)^2 \\
        &= \frac{1}{2^{2n}} \E_{y \sim \{0,1\}^n} \sum_{x\in\{0,1\}^n} \left(\sum_{s \in \hat{\mathsf{P}}_n^Z} \mathbf{P}_b(\Pi_{M_b}(s) \neq 0) \Tr(\mc U^\dagger(\ketbra{x}) \mc C(s) ) \Tr(s \ketbra{y})\right)^2 \\
        &= \frac{1}{2^{2n}} \E_{y \sim \{0,1\}^n} \sum_{x\in\{0,1\}^n} \sum_{s\in \hat{\mathsf{P}}_n^Z} \mathbf{P}_b(\Pi_{M_b}(s) \neq 0)^2 \Tr(\mc U^\dagger(\ketbra{x}) \mc C(s) )^2 \Tr(s \ketbra{y}))^2 \label{eq:orthogonality}\\
        & \leq \frac{1}{2^{2n}}  \sum_{x\in\{0,1\}^n} \sum_{s\in \hat{\mathsf{P}}_n^Z}  \mathbf{P}_b(\Pi_{M_b}(s) \neq 0)^2 \label{eq:order-1} \\
         &= \frac{1}{2^n}\sum_{s\in \hat{\mathsf{P}}_n^Z}  \mathbf{P}_b(\Pi_{M_b}(s) \neq 0)^2\\
         &= \frac{1}{2^n}\left[1+ \sum_{s\in \{\hat{\mathsf{P}}_n^Z-I^{\otimes n}\}}  \mathbf{P}_b(\Pi_{M_b}(s) \neq 0)^2\right]\\
         &\leq \frac{1}{2^n}\left[1+ (1-\gamma)^{d}  \sum_{s\in \{\hat{\mathsf{P}}_n^Z-I^{\otimes n}\}}  \mathbf{P}_b(\Pi_{M_b}(s) \neq 0)\right] \label{eq:weight-at-least-one}\\
         &= \frac{1}{2^n}  \left[1+(1-\gamma)^{d} \E_b |\mathsf{C}(\langle M_b\rangle)|\right]\\
         &\leq \frac{1}{2^n} \left[1 +  (1-\gamma)^{d} de^{3(1-\gamma)^dn}\right] \label{eq:final}
    \end{align}
    where we have used orthogonality in \cref{eq:orthogonality}, the fact that Pauli observables on states are bounded by $1$ \cref{eq:order-1}, the fact that Pauli operators with weight at least $1$ survive with probability at most $(1-\gamma)^d$ in \cref{eq:weight-at-least-one}, and \cref{lemma:num_surviving_paulis} in \cref{eq:final}. Therefore, we achieve anticoncentration when $d > O(\gamma^{-1} \log n)$.
\end{proof}

\bibliography{doms_refs,bibliography}

\end{document}